\documentclass{IEEEtran}
\usepackage{hyperref}
\usepackage{flushend}
\usepackage{cite,mathtools}
\usepackage{amsmath,amssymb,amsfonts,array}
\usepackage{enumerate}
\newtheorem{proposition}{Proposition}
\newtheorem{remark}{Remark}
\newtheorem{definition}{Definition}
\newtheorem{theorem}{Theorem}
\newtheorem{example}{Example}

\title{A Relation Between Network Computation and Functional Index Coding Problems}
\begin{document}

\author{
\authorblockN{Anindya Gupta and B. Sundar Rajan\\}
\authorblockA{Department of Electrical Communication Engineering, Indian Institute of Science, Bengaluru 560012, KA, India\\Email:~\{anindya.g, bsrajan\}@ece.iisc.ernet.in}
}

\maketitle
\begin{abstract}
In contrast to the network coding problem wherein the sinks in a network demand subsets of the source messages, in a network computation problem the sinks demand functions of the source messages. 
Similarly, in the functional index coding problem, the side information and demands of the clients include disjoint sets of functions of the information messages held by the transmitter instead of disjoint subsets of the messages, as is the case in the conventional index coding problem. It is known that any network coding problem can be transformed into an index coding problem and vice versa.
In this work, we establish a similar relationship between network computation problems and a class of functional index coding problems, viz., those in which only the demands of the clients include functions of messages. We show that any network computation problem can be converted into a functional index coding problem wherein some clients demand functions of messages and vice versa. We prove that a solution for a network computation problem exists if and only if a functional index code (of a specific length determined by the network computation problem) for a suitably constructed functional index coding problem exists. And, that a functional index coding problem admits a solution of a specified length if and only if a suitably constructed network computation problem admits a solution.
\end{abstract} 

\section{Introduction}

Conventional communication networks, like the Internet, ensure transfer of information generated at some nodes to others. It is known that network coding affords throughput gain over routing in such networks \cite{Ahls,Yeung,Alg,Fro,Ho,RLNC,Insuff}, and given a network and the demanded source messages at each sink, the network coding problem is to design a network code that maximizes the rate of information transfer from the source nodes to the sinks. But in some networks, like a sensor networks for environmental monitoring, nodes may be interested not in the messages generated by some other nodes but in one or more functions of these messages. Designing a network code that maximizes the frequency of target functions computation, called the \textit{computing capacity}, per network use at the sinks is known as the \textit{network computing} problem \cite{Appu}. This subsumes the network coding problem as a special case. Environmental monitoring in an industrial unit is an application of network computation where relevant parameter may include temperature and level of exhaust gases which may assist in preventing fire and poisoning due to toxic gases respectively. 

A simple way to perform network computation is to communicate all the messages relevant to the function required at each sink using either network coding or routing. This is not only highly inefficient in terms of bandwidth usage and power consumption but also undesirable in certain settings. For example, in an election, who voted whom is to be kept confidential but the sum total of votes received by each candidate is to be publicized. An efficient way is that function computation be performed \textit{in-network}, i.e., in a distributed manner. The intermediate nodes on the paths between the sources and the sinks perform network coding and communicate coded messages such that the sinks may compute their desired functions without having to know the value of the arguments. For example, if the maximum temperature is of interest in a sensor network, then each sensor node need only communicate the maximum of its own temperature reading and the message it receives on its incoming edges, which may be raw measurements or maximum temperatures of some sets of nodes, and communicate it to its neighbor on the path towards the sink node; the sink can then take maximum of the messages it receives to obtain the maximum temperature reading among all the nodes. Thus, network computation is performed in a distributed fashion without having to communicate measurement of each sensor node to the sink.

In \cite{GiriKr}, bounds on rate of computing symmetric functions (invariant to argument permutations), like minimum, maximum, mean, median and mode, of data collected by sensors in a wireless sensor network at a sink node were presented. The notion of min-cut bound for the network coding problem \cite{Yeung} was extended to the function computation problem in a directed acyclic network with multiple sources and one sink in \cite{Appu}. The case of directed acyclic network with multiple sources, multiple sinks and each sink demanding the sum of source messages was studied in \cite{Dey}; such a network is called a sum-network. Relation between linear solvability of multiple-unicast networks and sum-networks was established. Furthermore, insufficiency of scalar and vector linear network codes to achieve computing capacity for sum-networks was shown. Coding schemes for computation of arbitrary functions in directed acyclic network with multiple sources, multiple sinks and each sink demanding a function of source messages were presented in \cite{Dey2}. In \cite{Appu2}, routing capacity, linear coding capacity and nonlinear coding capacity for function computation in a multiple source single sink directed acyclic network were compared and depending upon the demanded functions and alphabet (field or ring), advantage of linear network coding over routing and nonlinear network coding over linear network coding was shown.

The index coding problem was introduced in \cite{BK} and finds potential commercial application in dissemination of popular multimedia content in wireless ad hoc networks and using digital video broadcast. An instance of the index coding problem comprises a source, which generates a finite number of messages, and a set of receivers. Each receiver knows a subset of messages, called the \textit{Has-set}, and demands another subset of messages, called the \textit{Want-set}. The objective is to make a minimum number of encoded transmissions over a noiseless broadcast channel such that the demands of all the clients are satisfied upon reception of the same. Graph theory has been extensively used to study this problem \cite{BBJK,LubStav,BCSI,OngHo}. The case where the \textit{Has-sets} includes linear combination of messages was studied in \cite{BCCSI} and \cite{ICCSI}; such a scenario may arise if some clients fail to receive some coded transmissions, possibly due to power outage, bad weather, intermittent signal reception, etc., and the transmitter will have to compute a new index code after every transmissions taking into account the updated caches, which may now include coded messages. 

The \textit{functional index coding problem} was recently proposed in \cite{FICP} as a generalization of the conventional index coding problem. In a functional index coding problem, the \textit{Has-} and \textit{Want-sets} of users may contain functions of messages rather than only a subset of messages as is the case in a conventional index coding problem. Thus, the conventional index coding problem and the scenarios studied in \cite{BCCSI} and \cite{ICCSI} are special cases of the functional index coding problem. In \cite{FICP}, bounds on codebook size were obtained and a graph coloring approach to obtaining a \textit{functional index code} was given. 


In \cite{ICRel}, it was shown that any network coding problem can be reduced to an index coding problem and that a network coding problem admits a linear solution if and only if a linear index code of a specific length (determined by the network coding problem) exists for the corresponding index coding problem. This relationship was extended to include nonlinear codes in \cite{ICRel2}. In \cite{ICRel2}, a method, similar to that given in \cite{ICRel}, to convert a network coding problem to an index coding problem was given. It was shown that a network code for the former problem exists if and only if an index code of a specific length (determined by the network coding problem) exists; methods to convert an index code to a network code and vice versa were also given.



\subsection{Contributions and Organization}


In this paper, we explore the relationship of network computation and a class of functional index coding problems. The contributions of this paper are as follows:
\begin{enumerate}
\item In Section~\ref{sec_fnc2fic}, we give a method to construct a functional index coding problem from a given network computation problem (Definition~\ref{nc_2_fic}). In the resulting functional index coding problem, only the \textit{Want-sets} of the clients include functions of messages and the \textit{Has-sets} are all subsets of the message set (like in the conventional index coding problem). 
\item We show that a network code to perform in-network computation of the functions demanded by the sinks exists if and only if the corresponding functional index coding problem admits a solution of a specific length, which is determined by the original network (Theorem~\ref{thm_1} in Section~\ref{sec_fnc2fic}). We give a method to convert a functional index code for the obtained functional index coding problem into a network code for the original network computation problem and vice versa. Thus, in order to obtain a network code for a given network computation problem, we convert it into a functional index coding problem, construct a functional index code, and then convert the functional index code to a network code for the given network computation problem.
\item In Section~\ref{sec_fic2fnc}, we show that any functional index coding problem with only \textit{Want-sets} containing functions of messages can be converted as a network computation problem (Definition~\ref{fic_2_fnc}). We prove that the functional index coding problem admits a solution of a specific length if and only if the network computation problem obtained from it admits a solution (Proposition~\ref{prop_2}). 
\end{enumerate}
In Section~\ref{sec_prelims}, relevant preliminaries of network computation and functional index coding problems are given. We conclude the paper with a summary of work presented and possible directions of future work in Section~\ref{sec_disc}. 


\section{Network Model} \label{sec_prelims}
A brief overview of network computation and functional index coding problems are presented in this section. A $q$-ary finite field is denoted by $\mathbb{F}_q$ and the set $\{1,2,\ldots,n\}$ is denoted by $[n]$, for some positive integer $n$. All vectors are row vectors. A realization of a random variable $Z$ taking value from $\mathbb{F}_q^n$, for some positive integer $n$, is denoted by $z$. A set containing one element $s$ is denoted by $\{s\}$.

\subsection{Network Computation}
A network is represented by a finite directed acyclic graph $\mathcal{N}=(V,\mathcal{E})$, where $V$ is the set of nodes and $\mathcal{E}=E\cup \tilde{E}$ is the set of directed error-free links (edges), where the edges in $E$ correspond to the links between the nodes in the network and the edges in $\tilde{E}$ correspond to the source messages generated in the network. The sets of incoming and outgoing links of a node $w\in V$ are denoted by $In(w)$ and $Out(w)$ respectively. For an edge $e=(u,v)\in E$ from a node $u$ to $v$, $u$ and $v$ are called, respectively, its tail and head and $In(e)=In(u)$, i.e., $In(e)$ is the set of edges which terminate at the node at which $e$ originates. The network may have multiple source nodes and each may generate multiple messages. The source messages are represented by tailless edges $\tilde{e}_k\in \tilde{E}$ that terminate at a source node. The total number of messages generated in the network is $K=|\tilde{E}|$ and are denoted by random variables $X_1,X_2,\ldots,X_K$, where, for every $k\in [K]$, $X_k$ is uniformly distributed over $\mathbb{F}_q^{n_k}$, for some positive integer $n_k$. Let $N_K=n_1+n_2+\ldots+n_K$. Let $X=(X_1,X_2,\ldots,X_K)$. A realization of $X$ is denoted by $\mathbf{x}=(x_1,x_2,\ldots,x_K)$, where $x_k\in \mathbb{F}_q^{n_k}$ for all $k\in [K]$ and $\mathbf{x} \in \mathbb{F}_q^{N_K}$. The capacity of a link $e\in E$ between nodes in the network is $n_e$ and $Y_e$ is the associated random variable, i.e., $y_e\in \mathbb{F}_q^{n_e}$. Note that for a source edge $\tilde{e}_k\in \tilde{E}$, the capacity is $n_k$ and the associated random variable is $Y_{\tilde{e}_k}=X_k$ respectively. The set of sink nodes is denoted by ${T}$. 
A sink node $t$ demands $n_t$ functions $G_t=\{g_{t,1}(X),g_{t,2}(X),\ldots,g_{t,n_t}(X)\}$, where $g_{t,i}:\mathbb{F}_q^{N_K}\rightarrow \mathbb{F}_q$ for all $i\in [n_t]$. Let $G_t(X)=(g_{t,1}(X),g_{t,2}(X),\ldots,g_{t,n_t}(X))$. For a realization $\mathbf{x}$ of $X$, $G_t(\mathbf{x})=(g_{t,1}(\mathbf{x}),g_{t,2}(\mathbf{x}),\ldots,g_{t,n_t}(\mathbf{x}))$ is the vector of function values sink $t$ wishes to compute. The demands are denoted by headless edges (not included in $\mathcal{E}$) originating at sink nodes. 

A network computation problem $\mathcal{F}(\mathcal{N}(V,\mathcal{E}),X,\{G_t:t\in {T}\})$ is specified by the underlying network, the source messages, and the demands of each sink.

A network code $\{f_e:e\in E\}\cup \{D_t:t\in {T}\}$ for a network computation problem $\mathcal{F}(\mathcal{N}(V,\mathcal{E}),X,\{G_t:t\in {T}\})$ is an assignment of a \textit{local encoding kernels} $f_e$ to each edge $e\in E$ and a decoding function $D_t$ to each sink $t\in {T}$. For any $e\in E$, $f_e$ takes in $(Y_{e'})_{e'\in In(e)}$ as input argument and outputs $Y_e$, i.e.,
\begin{align*}
f_e:(Y_{e'})_{e'\in In(e)}\longmapsto Y_e.
\end{align*}
For any sink $t\in {T}$, the decoding map $D_t$ takes in $(Y_{e'})_{e'\in In(t)})$ as input argument and outputs $G_t(X)$, i.e.,
\begin{align}\label{eq_dec_y}
D_t:(Y_{e'})_{e'\in In(t)}\longmapsto G_t(X)
\end{align}
for every realization $\mathbf{x}$ of $X$, $\mathbf{x}\in \mathbb{F}_q^{n_1+n_2+\ldots+n_K}$.

The \textit{global encoding kernels} $\{F_e:e\in \mathcal{E}\}$ and the decoding functions $\{D_t:t\in {T}\}$ give an alternate description of a network code. For any edge $e$, $F_e$ maps $X$ to $Y_e$ (and thus, the distribution of $Y_e$ depends upon the network code). Given the local encoding kernels, the global encoding kernels for each edge can be defined by induction on an ancestral ordering of the edges in the graph as follows \cite[Sec.~II]{ICRel2}. For every tailless edge $\tilde{e}_k\in\tilde{E}$ denoting the source message $X_k$, let $F_{\tilde{e}_k}(X)=X_k$ be the global encoding kernel. Then, for any edge $e\in E$, 
\begin{align} \label{eq_f_2_F}
F_e(X)=f_e\left((F_{e'}(X))_{e'\in In(e)}\right).
\end{align}
Also, by \eqref{eq_dec_y},
\begin{align} \label{eq_dec_F}
D_t\left((F_{e'}(X))_{e'\in In(t)}\right)=G_t(X),
\end{align}
for every realization $\mathbf{x}$ of $X$, $\mathbf{x}\in \mathbb{F}_q^{N_K}$.

\subsection{Functional Index Coding}
An instance $\mathcal{I}(Z,\mathcal{R})$ of a functional index coding problem comprises
\begin{enumerate}
\item a transmitter equipped with the message vector $Z=(Z_1,Z_2,\ldots,Z_K)$, where, for every $k\in [K]$, $Z_k$ is uniformly distribute over $\mathbb{F}_q^{n_k}$ for some positive integer $n_k$ and
\item a set of clients or receivers, $\mathcal{R}=\{R_1,R_2,\ldots,R_{|\mathcal{R}|}\}$, where $R_i=(H_i,W_i)$ for all $R_i\in\mathcal{R}$. For any receiver $R_i$, $H_i=\{ h_{i,1}(Z),h_{i,2}(Z),\ldots ,h_{i,|H_i|}(Z)\}$ and $W_i=\{w_{i,1}(Z),w_{i,2}(Z),\ldots ,w_{i,|W_i|}(Z)\}$ are the \textit{Has-} and \textit{Want-sets} respectively, where $h_{i,j},w_{i,l}:\mathbb{F}_q^{N_K} \rightarrow \mathbb{F}_q$ for $1\leqslant j\leqslant |H_i|$ and $1\leqslant l\leqslant |W_i|$ $(N_K=n_1+\ldots+n_K)$.
\end{enumerate}

Let $H_i(Z)=(h_{i,1}(Z),h_{i,2}(Z),\ldots ,h_{i,|H_i|}(Z))$ and $W_i(Z)=(w_{i,1}(Z),w_{i,2}(Z),\ldots ,w_{i,|W_i|}(Z))$. For a realization $\mathbf{z}=(z_1,z_2,\ldots,z_K)$ of $Z$ at the transmitter, $H_i(\mathbf{z})=(h_{i,1}(\mathbf{z}),h_{i,2}(\mathbf{z}),\ldots ,h_{i,|H_i|}(\mathbf{z}))$ is the \textit{Has-value} known to $R_i$ and $W_i(\mathbf{z})=(w_{i,1}(\mathbf{z}),w_{i,2}(\mathbf{z}),\ldots ,w_{i,|W_i|}(\mathbf{z}))$ as the \textit{Want-value} $R_i$ is interested in computing.

The conventional index coding problem correspond to the case in which all the \textit{Has-} and \textit{Want-set} are subsets of $Z$. The case in which the \textit{Has-sets} of the clients include linear combinations of messages and the \textit{Want-sets} are subsets of $Z$ was considered in \cite{BCCSI} and \cite{ICCSI}.

A functional index code for a given $\mathcal{I}(Z,\mathcal{R})$ comprises an encoding map
\begin{align*}
\mathcal{M}:\mathbb{F}_q^{N_K}\longrightarrow \mathbb{F}_q^{\ell}
\end{align*}
and a decoding function $\hat{D}_{R_i}$ for every receiver $R_i\in \mathcal{R}$, where 
\begin{align*}
\hat{D}_{R_i}: \mathbb{F}_q^\ell \times \mathbb{F}_q^{|H_i|}\longrightarrow \mathbb{F}_q^{|W_i|}
\end{align*}
and satisfies 
\begin{align} \label{eq_Dhat}
\hat{D}_{R_i}(\mathcal{M}(Z),H_i(Z))=W_i(Z)
\end{align}
for every realization $\mathbf{z}$ of $Z$, $\mathbf{z}\in \mathbb{F}_q^{N_K}$.

Here $\ell$ is the \textit{length} of the functional index code. Let $\mathcal{B_M}=\{\mathcal{M}(\mathbf{z}):\mathbf{z}\in \mathbb{F}_q^{N_K}\}\subseteq \mathbb{F}_q^{\ell}$ denote the \textit{codebook} of the code defined by $\mathcal{M}$. The transmitter broadcasts $\ell$ length codewords, $\mathcal{M}(\mathbf{z})$, and the receivers use their respective decoding functions to obtain the values of their desired functions. The objective of code design is to minimize $\ell$ so that maximum throughput gain is achieved.

\section{From Network Computation Problem to Functional Index Coding Problem}\label{sec_fnc2fic}
In this section, we obtain a functional index coding problem from a network computation problem and show that a given network computation problem is feasible if and only if the corresponding functional index coding problem is. The proof technique followed is similar to that of \cite[Th.~1]{ICRel2}.

\begin{definition}
\label{nc_2_fic}
Let $\mathcal{F}(\mathcal{N}(V,\tilde{E}\cup E),X,\{G_t:t\in {T}\})$ be a network computation problem with $K$ messages $X_1,X_2,\ldots,X_K$, where, for every $k\in [K]$, $X_k$ is uniformly distributed over $\mathbb{F}_q^{n_k}$ for some positive integer $n_k$, the capacity of an edge $e\in E$ is $n_e$ for some positive integer $n_e$, and, every sink $t\in T$ demands a set of $n_t$ functions, $G_t=\{g_{t,1}(X),g_{t,2}(X),\ldots,g_{t,n_t}(X)\}$ for some positive integer $n_t$. A functional index coding problem $\mathcal{I_F}(Z,\mathcal{R})$ can be obtained as follows:
\begin{enumerate}
\item The transmitter has access to the message vector $Z=(\hat{X},\hat{Y})$, where $\hat{X}=(\hat{X}_k)_{k\in [K]}$ and $\hat{Y}=(\hat{Y}_e)_{e\in E}$. A message $\hat{X}_k$, $k\in [K]$, is uniformly distributed over $\mathbb{F}_q^{n_k}$ and a message $\hat{Y}_e$, $e\in E$, is uniformly distributed over $\mathbb{F}_q^{n_e}$ for every $e\in E$. The quantities $n_k$, $k\in [K]$, and $n_e$, $e\in E$, are specified by the network computation problem $\mathcal{F}$.
\item The set $\mathcal{R}=R_E\cup R_T\cup \{R_{\text{all}}\}$ of $|E|+|T|+1$ clients is defined as follows:
\begin{itemize}
\item $R_E=\{R_e:e\in E\}$, where the \textit{Has-} and \textit{Want-sets} of $R_e$ are $H_e=\{\hat{Y}_{e'}:e'\in In(e)\}$ and $W_e=\{\hat{Y}_e\}$ respectively for any $e\in E$.
\item $R_T=\{R_t:t\in T\}$, where the \textit{Has-} and \textit{Want-sets} of $R_t$ are $H_t=\{\hat{Y}_e:e\in In(t)\}$ and $W_t=\{g_{t,1}(\hat{X}),g_{t,2}(\hat{X}),\ldots,g_{t,n_t}(\hat{X})\}$ respectively for any $t\in T$.
\item The \textit{Has-set} of the client $R_{\text{all}}$ is $\{\hat{X}_k:k\in [K]\}$ and its \textit{Want-sets} is $\{\hat{Y}_e:e\in E\}$.
\end{itemize}
\end{enumerate}
\end{definition}

Note that in the functional index coding problem obtained from a network computation problem using the above definition
\begin{enumerate}[a)]
\item the \textit{Has-set} of every receiver will be a subset of the messages $\{\hat{X}_1,\hat{X}_2,\ldots,\hat{X}_K,\hat{Y}_{e_1},\hat{Y}_{e_2},\ldots,\hat{Y}_{e_{|E|}}\}$ and
\item only the \textit{Want-sets} of the receivers in $R_T$, i.e., the receivers corresponding to sinks in the starting network computation problem, will contain functions of messages.
\end{enumerate} 

\begin{theorem} \label{thm_1}
For a given network computation problem $\mathcal{F}(\mathcal{N}(V,\tilde{E}\cup E),X,\{G_t:t\in {T}\})$, let $\mathcal{I_F}(Z,\mathcal{R})$ be the corresponding functional index coding problem constructed using Definition~\ref{nc_2_fic}. Then, a network code for $\mathcal{F}$ exists if and only if a functional index code of length $\sum_{e\in E}n_e$ exists for $\mathcal{I_F}$. 
\end{theorem}
\begin{proof}
Let $N_K=\sum_{k\in [K]}n_k$ and $N_E=\sum_{e\in E}n_e$. To prove the theorem, we will show that any network code for $\mathcal{F}$ can be converted into a functional index code of length $N_E$ for $\mathcal{I_F}$ and vice versa.

\textit{Converting a network code into a functional index code:} Suppose there is a network code for $\mathcal{F}$ and let $\{f_e:e\in E\}$, $\{F_e:e\in \tilde{E}\}$, and $\{D_t:t\in T\}$ be its local encoding kernels, global encoding kernels, and decoding functions respectively. 


Define a map
\begin{align*}
\mathcal{M}:\mathbb{F}_q^{N_K + N_E}\longrightarrow \mathbb{F}_q^{N_E}
\end{align*}
by
\begin{align} \label{eq_map}
\mathcal{M}(Z)=(M_e(Z))_{e\in E}=(\hat{Y}_e+F_e(\hat{X}))_{e\in E},
\end{align}
where $M_e:Z=(\hat{X},\hat{Y})\mapsto \hat{Y}_e+F_e(\hat{X})$, for any $e\in E$. We will verify that broadcasting $(M_e(Z))_{e\in E}$ satisfies all the receivers of $\mathcal{I_F}$ by constructing a decoding map for each receiver using the local and global encoding kernels and the decoding functions specified by the network code. The length of this code is $\sum_{e\in E}n_e=N_E$.

\begin{enumerate}

\item A receiver $R_e\in R_E$ knows $H_{e}=\{\hat{Y}_{e'}:e'\in In(e)\}$ and wants $W_{e}=\{\hat{Y}_e\}$. Note that $H_{e}(Z)=(\hat{Y}_{e'})_{e'\in In(e)}$ and $W_{e}(Z)=\hat{Y}_e$. The decoding map $\hat{D}_{R_e}$ is obtained as follows:
\begin{itemize}
\item Compute $F_{e'}(\hat{X})=M_{e'}(Z)-\hat{Y}_{e'}$ for each $e'\in In(e)$ using the known information $\hat{Y}_{e'}$ and the broadcast message $M_{e'}(Z)$.
\item Then, using \eqref{eq_f_2_F}, compute
\begin{align*}
F_e(\hat{X})=f_e\left((F_{e'}(\hat{X}))_{e'\in In(e)}\right).
\end{align*}
\item The demanded message is $\hat{Y}_e=M_e(Z)-F_e(\hat{X})$.
\end{itemize}
Thus, for a receiver $R_e\in R_E$, the map
\begin{align} \label{eq_dh_E} \nonumber
\hat{D}_{R_e}\left(\;(M_e(Z))_{e\in E}\;,\;(\hat{Y}_{e'})_{e'\in In(e)}\;\right)=\quad & 
\\ M_e(Z)-f_e \left( (M_{e'}(Z)-\hat{Y}_{e'})_{e'\in In(e)} \right)&=\hat{Y}_e
\end{align} 
satisfies \eqref{eq_Dhat} and hence is a decoding function for $R_e$. 

\item A receiver $R_t\in R_T$ knows $H_{t}=\{\hat{Y}_{e'}:e'\in In(t)\}$ and wants $W_{t}=\{g_{t,1}(\hat{X}),g_{t,2}(\hat{X}),\ldots,g_{t,n_t}(\hat{X})\}$. Note that $H_{t}(Z)=(\hat{Y}_{e'})_{e'\in In(t)}$ and $W_{t}(Z)=G_t(\hat{X})$. The decoding map $\hat{D}_{R_t}$ is obtained as follows:
\begin{itemize}
\item Compute $F_{e'}(\hat{X})=M_{e'}(Z)-\hat{Y}_{e'}$ for each $e'\in In(t)$ using the known information $\hat{Y}_{e'}$ and the broadcast message $M_{e'}(Z)$.
\item Then, by \eqref{eq_dec_F}
\begin{align*}
G_t(\hat{X})=D_t\left((F_{e'}(\hat{X}))_{e'\in In(t)}\right).
\end{align*}
\end{itemize}
Thus, for a receiver $R_t\in R_T$, the map
\begin{align} \label{eq_dh_T} \nonumber
\hat{D}_{R_t}\left(\;(M_e(Z))_{e\in E}\,,\,(\hat{Y}_{e'})_{e'\in In(t)}\;\right)= \qquad\;\;\; &
\\D_t\left(\;(M_{e'}(Z)-\hat{Y}_{e'})_{e'\in In(t)}\;\right)=\;G_t&(\hat{X})
\end{align}
satisfies \eqref{eq_Dhat} and hence is a decoding function for $R_t$.

\item Since $R_{\text{all}}$ knows all $\hat{X}_k$, $k\in [K]$, and demands all $\hat{Y}_e$, $e\in E$, the decoding function for receiver $R_{\text{all}}$ is 
\begin{align} \label{eq_dh_all} \nonumber
\hat{D}_{R_{\text{all}}}\left((M_e(Z))_{e\in E},\hat{X} \right)&= \left( M_e(Z)-F_e(\hat{X}) \right)_{e\in E}\\ &=(\hat{Y}_e)_{e\in E}.
\end{align} 

\end{enumerate}

Thus, any network code for $\mathcal{F}$ can be converted into a functional index code of length $N_E$ for $\mathcal{I_F}$ using \eqref{eq_map}, \eqref{eq_dh_E}, \eqref{eq_dh_T}, and \eqref{eq_dh_all}.

We now prove the converse.

\textit{Converting a functional index code into a network code:} Assume that a functional index of length $N_E$ is given for $\mathcal{I_F}(Z,\mathcal{R})$, i.e., an encoding map $\mathcal{M}$ and decoding functions $\{\hat{D}_{R_e}:R_e\in R_E\}\cup \{\hat{D}_{R_t}:R_t\in R_T\}\cup \{\hat{D}_{R_{\text{all}}}\}$ are specified. We will construct a network code for the network computation problem $\mathcal{F}(\mathcal{N}(V,\tilde{E}\cup E),X,\{G_t:t\in {T}\})$, i.e., specify local encoding kernels for edges in $E$ and decoding functions for sinks in $T$, using the given functional index code. Let $\mathcal{B_M}\subseteq \mathbb{F}_q^{N_E}$ denote the codebook.

Since $R_{\text{all}}$ can decode its demanded messages, it must be true that $\mathcal{M}(\hat{\mathbf{x}},\hat{\mathbf{y}})\neq \mathcal{M}(\hat{\mathbf{x}},\hat{\mathbf{y}}')$ for every $\hat{\mathbf{x}}\in \mathbb{F}_q^{N_K}$ and every $\hat{\mathbf{y}},\hat{\mathbf{y}}'\in \mathbb{F}_q^{N_E}$ such that $\hat{\mathbf{y}}\neq \hat{\mathbf{y}}'$ (otherwise, if $\mathcal{M}(\hat{\mathbf{x}},\hat{\mathbf{y}})\!=\! \mathcal{M}(\hat{\mathbf{x}},\hat{\mathbf{y}}')\!=\!\mathbf{m}$ (say), then, given $\hat{\mathbf{x}}$ and the broadcast codeword $\mathbf{m}$, $R_{\text{all}}$ will not be able to distinguish between $\hat{\mathbf{y}}$ and $\hat{\mathbf{y}}'$).\footnote{For every receiver $R$ in a functional index coding problem $\mathcal{I}(Z,\mathcal{R})$, if, for some $\mathbf{z}\neq \mathbf{z}'$, $H_R(\mathbf{z})=H_R(\mathbf{z}')$ and $W_R(\mathbf{z})\neq W_R(\mathbf{z}')$, then an encoding map $\mathcal{M}$ must satisfy $\mathcal{M}(\mathbf{z})\neq \mathcal{M}(\mathbf{z}')$, otherwise the demands of $R$ will not be met \cite[Prop.~1]{FICP}. This constraint on the encoding map is referred to as the \textit{generalized exclusive law} \cite[Def.~3]{FICP}.} Thus, for each $\hat{\mathbf{x}}\in \mathbb{F}_q^{N_K}$, the elements of $B_{\hat{\mathbf{x}}} = \{\mathcal{M}(\hat{\mathbf{x}},\hat{\mathbf{y}}):\hat{\mathbf{y}}\in \mathbb{F}_q^{N_E}\}$ are all distinct and hence $|B_{\hat{\mathbf{x}}}|=q^{N_E}$. Consequently, $B_{\hat{\mathbf{x}}} =\mathbb{F}_q^{N_E}$ (since $B_{\hat{\mathbf{x}}}\subseteq \mathcal{B_M}\subseteq \mathbb{F}_q^{N_E}$ and $|B_{\hat{\mathbf{x}}}|=q^{N_E}$). Thus, given any $\mathbf{m}\in \mathcal{B_M}$ and any $\hat{\mathbf{x}}\in \mathbb{F}_q^{N_K}$, there exists a unique realization $\hat{\mathbf{y}}_{\hat{\mathbf{x}}}$ of $\hat{Y}$ such that $\mathcal{M}(\hat{\mathbf{x}},\hat{\mathbf{y}}_{\hat{\mathbf{x}}})=\mathbf{m}$. Specifically, let $\mathbf{m}=\mathbf{0}$, where $\mathbf{0}$ is the all zero $N_E$-tuple, and $\mathcal{A}_{0}=\{(\hat{\mathbf{x}},\hat{\mathbf{y}}_{\hat{\mathbf{x}}}):\hat{\mathbf{x}}\in \mathbb{F}_q^{N_K},\mathcal{M}(\hat{\mathbf{x}},\hat{\mathbf{y}}_{\hat{\mathbf{x}}})=\mathbf{0}\}$. Note that $|\mathcal{A}_0|=q^{N_K}$ (since for each $\hat{\mathbf{x}}\in \mathbb{F}_q^{N_K}$, there exists a unique $\hat{\mathbf{y}}_{\hat{\mathbf{x}}}$ with $\mathcal{M}(\hat{\mathbf{x}},\hat{\mathbf{y}}_{\hat{\mathbf{x}}})=\mathbf{0}$). 
Let $\hat{\mathbf{y}}_{\hat{\mathbf{x}}}=(\hat{y}_{e,\hat{\mathbf{x}}})_{e\in E}$. Then, by \eqref{eq_Dhat}, for every $(\hat{\mathbf{x}},\hat{\mathbf{y}}_{\hat{\mathbf{x}}})\in \mathcal{A}_0$
\begin{align} \label{eq_dhat_e} \nonumber
\hat{D}_{R_e}\left(\mathcal{M}(\hat{\mathbf{x}},\hat{\mathbf{y}}_{\hat{\mathbf{x}}}), H_e(\hat{\mathbf{x}},\hat{\mathbf{y}}_{\hat{\mathbf{x}}}) \right)&= \hat{D}_{R_e}\left( \mathbf{0}, (\hat{y}_{e',\hat{\mathbf{x}}})_{e'\in In(e)} \right)\\&= \hat{y}_{e,\hat{\mathbf{x}}}
\end{align}
for every receiver $R_e\in R_E$ and 
\begin{align} \label{eq_dhat_t} \nonumber
\hat{D}_{R_t}\left(\mathcal{M}(\hat{\mathbf{x}},\hat{\mathbf{y}}_{\hat{\mathbf{x}}}), H_t(\hat{\mathbf{x}},\hat{\mathbf{y}}_{\hat{\mathbf{x}}}) \right)&= \hat{D}_{R_t}\left( \mathbf{0}, (\hat{y}_{e',\hat{\mathbf{x}}})_{e'\in In(t)} \right)\\&= G_t(\hat{\mathbf{x}})
\end{align}
for every receiver $R_t\in R_T$, and $\hat{D}_{R_{\text{all}}}\left(\mathcal{M}(\hat{\mathbf{x}},\hat{\mathbf{y}}_{\hat{\mathbf{x}}}),\hat{\mathbf{x}}\right)=\hat{D}_{R_{\text{all}}}\left(\mathbf{0},\hat{\mathbf{x}}\right)=\hat{\mathbf{y}}_{\hat{\mathbf{x}}}$.

For an edge $e\in E$, define a map
\begin{align*}
f_e:\mathbb{F}_q^{\sum_{e'\in In(e)}n_{e'}}\longrightarrow \mathbb{F}_q^{n_e}
\end{align*}
by 
\begin{align} \label{eq_dhre2fe}
f_e\left((Y_{e'})_{e'\in In(e)}\right)=\hat{D}_{R_e}\left(\mathbf{0},(Y_{e'})_{e'\in In(e)}\right)
\end{align}
and for a sink $t\in T$, define a map
\begin{align*}
D_t:\mathbb{F}_q^{\sum_{e'\in In(t)}n_{e'}}\longrightarrow \mathbb{F}_q^{n_t}
\end{align*}
by
\begin{align} \label{eq_dhrt2dt}
D_t\left((X_{e'})_{e'\in In(t)}\right)=\hat{D}_{R_t}\left(\mathbf{0},(Y_{e'})_{e'\in In(t)}\right).
\end{align}
We will verify that $\{f_e:e\in E\}\cup\{D_t:t\in T\}$ is a network code for $\mathcal{F}$.

Let $\mathbf{x}$ be an arbitrary realization of $X$ and $(\hat{\mathbf{x}},\hat{\mathbf{y}}_{\hat{\mathbf{x}}})$ be the element of $\mathcal{A}_0$ such that $\mathbf{x}=\hat{\mathbf{x}}$. 

Recall that a tailless edge $\tilde{e}_k\in \tilde{E}$, $k\in [K]$, represents a source message $X_k$ and that $Y_{\tilde{e}}=X_k$. First we will show that $\mathbf{y}=\hat{\mathbf{y}}_{\hat{\mathbf{x}}}$, i.e., $y_e=\hat{y}_{e,\hat{\mathbf{x}}}$ for every $e\in E$, by induction on an ancestral ordering of edges in $E$ as follows. Let an edge $e\in E$ originate at a source node such that $In(e)\subseteq \tilde{E}$ (or $In(e)\cap E=\emptyset$) (base case of induction). Then, by \eqref{eq_dhre2fe} and \eqref{eq_dhat_e},
\begin{align*} 
y_e=f_e\left((y_{e'})_{e'\in In(e)}\right)=\hat{D}_{R_e}\left(\mathbf{0},(y_{e'})_{e'\in In(e)}\right)=\hat{y}_{e,\hat{\mathbf{x}}}.
\end{align*}
Now let for some edge $e\in E$, every $e'\in In(e)$ satisfies $y_{e'}=\hat{y}_{e',\hat{\mathbf{x}}}$ (induction hypothesis). Then, again by \eqref{eq_dhre2fe} and \eqref{eq_dhat_e},
\begin{align*} 
y_e=f_e\left((y_{e'})_{e'\in In(e)}\right)=\hat{D}_{R_e}\left(\mathbf{0},(y_{e'})_{e'\in In(e)}\right)=\hat{y}_{e,\hat{\mathbf{x}}}. 
\end{align*}
Thus, $y_e=\hat{y}_{e,\hat{\mathbf{x}}}$ for every $e\in E$.
%

Now consider any $t\in T$. Then, by \eqref{eq_dhrt2dt} and \eqref{eq_dhat_t},
\begin{align*} 
D_t\left( (y_{e'})_{e'\in In(t)} \right)&=\hat{D}_{R_t}\left( \mathbf{0},(y_{e'})_{e'\in In(t)} \right)
\\ \nonumber &=\hat{D}_{R_t}\left( \mathbf{0},(\hat{y}_{e',\hat{\mathbf{x}}})_{e'\in In(t)} \right)=G_t(\hat{\mathbf{x}})
\\ & = G_t(\mathbf{x}),
\end{align*}
which is the required function value at $t$.

Thus, \eqref{eq_dhre2fe} and \eqref{eq_dhrt2dt} define a network code for $\mathcal{F}$.
\end{proof}

\begin{remark}
In the proof of the converse of Theorem~1, we chose $\mathbf{m}$ to be $\mathbf{0}$ to define a network code (\eqref{eq_dhre2fe} and \eqref{eq_dhrt2dt}). Instead of $\mathbf{0}$, $\mathbf{m}$ can be chosen to be any element of $\mathbb{F}_q^{N_E}$ to define a network code. Thus, for any $\mathbf{m}\in\mathbb{F}_q^{N_E}$,
\begin{align*}
f_e\left((Y_{e'})_{e'\in In(e)}\right)&=\hat{D}_{R_e}\left(\mathbf{m},(Y_{e'})_{e'\in In(e)}\right)
\\
D_t\left((X_{e'})_{e'\in In(t)}\right)&=\hat{D}_{R_t}\left(\mathbf{m},(Y_{e'})_{e'\in In(t)}\right)
\end{align*}
is a network code.
\end{remark}

\begin{example}
Consider the network computation problem given in Fig.~\ref{fig_fnc1}. There are $11$ source nodes each generating a $10$-bit long message and there is only one sink which wants to compute the maximum among the decimal equivalents of messages. All dashed and solid edges have capacity $10$ bits $(N_E=130)$. By $\max$, we mean the maximum of the decimal equivalent. The local encoding kernels are given adjacent to the edges.
\begin{figure}
\centering
\includegraphics[scale=0.5]{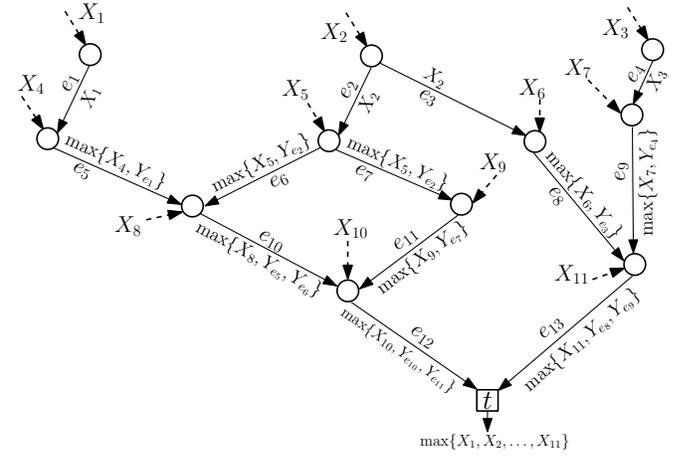}
\caption{A function computation problem.}
\label{fig_fnc1}
\vspace{-10pt}
\end{figure}
The local and global encoding kernels of edges $e_1,e_2,\ldots,e_{13}$ are given in Table~\ref{tab_LEKGEK} below. 
\begin{table}[h] 
\caption{}
\label{tab_LEKGEK}
\scriptsize
\centering
\renewcommand{\tabcolsep}{4pt}
\renewcommand{\arraystretch}{1.4}
\begin{tabular}{|c|c|c|}
\hline
$Y_e$ & $f_e\left((Y_{e'})_{e'\in In(e)}\right)$& $F_e(X_1,\ldots,X_{11})$ \\ \hline \hline
$Y_{e_1}$ & $X_1$& $X_1$ \\ \hline
$Y_{e_2}$ & $X_2$& $X_2$ \\ \hline
$Y_{e_3}$ & $X_2$& $X_2$ \\ \hline
$Y_{e_4}$ & $X_3$& $X_3$ \\ \hline
$Y_{e_5}$ & $\max\{Y_{e_1},X_4\}$ & $\max\{X_1,X_4\}$\\ \hline
$Y_{e_6}$ & $\max\{Y_{e_2},X_5\}$ & $\max\{X_2,X_5\}$\\ \hline
$Y_{e_7}$ & $\max\{Y_{e_2},X_5\}$ & $\max\{X_2,X_5\}$\\ \hline
$Y_{e_8}$ & $\max\{Y_{e_3},X_6\}$ & $\max\{X_2,X_6\}$\\ \hline
$Y_{e_9}$ & $\max\{Y_{e_4},X_7\}$ & $\max\{X_3,X_7\}$\\ \hline
$Y_{e_{10}}$ & $\max\{Y_{e_5},Y_{e_6},X_8\}$&$\max\{X_1,X_2,X_4,X_5,X_8\}$\\ \hline
$Y_{e_{11}}$ & $\max\{Y_{e_7},X_9\}$& $\max\{X_2,X_5,X_9\}$\\ \hline
$Y_{e_{12}}$ & $\max\{Y_{e_{10}},Y_{e_{11}},X_{10}\}$& $\max\{X_1,X_2,X_4,X_5,X_8,X_9,X_{10}\}$\\ \hline
$Y_{e_{13}}$ & $\max\{Y_{e_8},Y_{e_9},X_{11}\}$& $\max\{X_2,X_3,X_6,X_7,X_{11}\}$\\ \hline
\end{tabular}
\end{table}

\noindent The decoding function for sink $t$ is 
\begin{align*}
D_t(Y_{e_{12}},Y_{e_{13}})=\max\{Y_{e_{12}},Y_{e_{13}}\}=\max\{X_1,X_2,\ldots,X_{11}\}.
\end{align*}

The corresponding functional index coding problem obtained using Definition~\ref{nc_2_fic} is given by the first three columns of Table~\ref{tab_decmap1} at the top of the following page.
\begin{table*}[t] 
\scriptsize
\centering
\renewcommand{\tabcolsep}{1pt}
\renewcommand{\arraystretch}{1.4}
\begin{minipage}[t]{.6\linewidth}
\caption{} \label{tab_decmap1}  \vspace{-8pt} 
\begin{tabular}[t]{|>{\centering\arraybackslash}m{0.25in}|>{\centering\arraybackslash}m{0.75in}|>{\centering\arraybackslash}m{0.9in}|>{\centering\arraybackslash}m{2.2in}|}
\hline
\textit{Client} &\textit{Has-set} &\textit{Want-set}& \textit{Decoding Map} \\ \hline \hline
$R_{e_1}$ & $\hat{X}_1$ & $\hat{Y}_{e_1}$ & $M_{e_1}+\hat{X}_1$ \\ \hline
$R_{e_2}$ & $\hat{X}_2$ & $\hat{Y}_{e_2}$ & $M_{e_2}+\hat{X}_2$ \\ \hline
$R_{e_3}$ & $\hat{X}_2$ & $\hat{Y}_{e_3}$ & $M_{e_3}+\hat{X}_2$ \\ \hline
$R_{e_4}$ & $\hat{X}_3$ & $\hat{Y}_{e_4}$ & $M_{e_4}+\hat{X}_3$ \\ \hline
$R_{e_5}$ & $\hat{X}_4,\hat{Y}_{e_1}$ & $\hat{Y}_{e_5}$ & $M_{e_5}+\max\{M_{e_1}+\hat{Y}_{e_1},\hat{X}_4\}$ \\ \hline
$R_{e_6}$ & $\hat{X}_5,\hat{Y}_{e_2}$ & $\hat{Y}_{e_6}$ & $M_{e_6}+\max\{M_{e_2}+\hat{Y}_{e_2},\hat{X}_5\}$ \\ \hline
$R_{e_7}$ & $\hat{X}_5,\hat{Y}_{e_2}$ & $\hat{Y}_{e_7}$ & $M_{e_7}+\max\{M_{e_2}+\hat{Y}_{e_2},\hat{X}_5\}$ \\ \hline
$R_{e_8}$ & $\hat{X}_6,\hat{Y}_{e_3}$ & $\hat{Y}_{e_8}$ & $M_{e_8}+\max\{M_{e_3}+\hat{Y}_{e_3},\hat{X}_6\}$ \\ \hline
$R_{e_9}$ & $\hat{X}_7,\hat{Y}_{e_4}$ & $\hat{Y}_{e_9}$ & $M_{e_9}+\max\{M_{e_4}+\hat{Y}_{e_4},\hat{X}_7\}$ \\ \hline
$R_{e_{10}}$ & $\hat{X}_8,\hat{Y}_{e_5},\hat{Y}_{e_6}$ & $\hat{Y}_{e_{10}}$ & $M_{e_{10}}+\max\{M_{e_5}+\hat{Y}_{e_5},M_{e_6}+\hat{Y}_{e_6},\hat{X}_8\}$ \\ \hline
$R_{e_{11}}$ & $\hat{X}_9,\hat{Y}_{e_7}$ & $\hat{Y}_{e_{11}}$ & $M_{e_{11}}+\max\{M_{e_7}+\hat{Y}_{e_7},\hat{X}_9\}$ \\ \hline
$R_{e_{12}}$ & $\hat{X}_{10},\hat{Y}_{e_{10}},\hat{Y}_{e_{11}}$ & $\hat{Y}_{e_{12}}$ & $M_{e_{12}}+\max\{M_{e_{10}}+\hat{Y}_{e_{10}},M_{e_{11}}+\hat{Y}_{e_{11}},\hat{X}_{10}\}$ \\ \hline
$R_{e_{13}}$ & $\hat{X}_{11},\hat{Y}_{e_8},\hat{Y}_{e_9}$ & $\hat{Y}_{e_{13}}$ & $M_{e_{13}}+\max\{M_{e_8}+\hat{Y}_{e_8},M_{e_9}+\hat{Y}_{e_9},\hat{X}_{11}\}$ \\ \hline
$R_t$ & $\hat{Y}_{e_{12}},\hat{Y}_{e_{13}}$ & $\max\{\hat{X}_1,\ldots,\hat{X}_{11}\}$ & $\max\{M_{e_{12}}+\hat{Y}_{e_{12}},M_{e_{13}}+\hat{Y}_{e_{13}}\}$ \\ 
\end{tabular}

\end{minipage}%
\hfill
\begin{minipage}[t]{.4\linewidth}
\centering
\caption{} \label{tab_M_e1} \vspace{-8pt}
\begin{tabular}[t]{|c|c|}
\hline
$M_e(\hat{X},\hat{Y})$ & $\hat{Y}_e+F_e(\hat{X})$ \\ \hline \hline
$M_{e_1}$ & $\hat{Y}_{e_1}+\hat{X}_1$ \\ \hline 
$M_{e_2}$ & $\hat{Y}_{e_2}+\hat{X}_2$ \\ \hline 
$M_{e_3}$ & $\hat{Y}_{e_3}+\hat{X}_2$ \\ \hline 
$M_{e_4}$ & $\hat{Y}_{e_4}+\hat{X}_3$ \\ \hline 
$M_{e_5}$ & $\hat{Y}_{e_5}+\max\{\hat{X}_1,\hat{X}_4\}$ \\ \hline 
$M_{e_6}$ & $\hat{Y}_{e_6}+\max\{\hat{X}_2,\hat{X}_5\}$ \\ \hline 
$M_{e_7}$ & $\hat{Y}_{e_7}+\max\{\hat{X}_2,\hat{X}_5\}$ \\ \hline 
$M_{e_8}$ & $\hat{Y}_{e_8}+\max\{\hat{X}_2,\hat{X}_6\}$ \\ \hline 
$M_{e_9}$ & $\hat{Y}_{e_9}+\max\{\hat{X}_3,\hat{X}_7\}$ \\ \hline 
$M_{e_{10}}$ & $\hat{Y}_{e_{10}}+\max\{\hat{X}_1,\hat{X}_2,\hat{X}_4,\hat{X}_5,\hat{X}_8\}$ \\ \hline 
$M_{e_{11}}$ & $\hat{Y}_{e_{11}}+\max\{\hat{X}_2,\hat{X}_5,\hat{X}_9\}$ \\ \hline 
$M_{e_{12}}$ & $\hat{Y}_{e_{12}}+\max\{\hat{X}_1,\hat{X}_2,\hat{X}_4,\hat{X}_5,\hat{X}_8,\hat{X}_9,\hat{X}_{10}\}$ \\ \hline 
$M_{e_{13}}$ & $\hat{Y}_{e_{13}}+\max\{\hat{X}_2,\hat{X}_3,\hat{X}_6,\hat{X}_7,\hat{X}_{11}\}$ \\ \hline 
\end{tabular}
\end{minipage}%
\\
\begin{minipage}[t]{\linewidth}
\begin{tabular}{|>{\centering\arraybackslash}m{0.25in}|>{\centering\arraybackslash}m{0.75in}|>{\centering\arraybackslash}m{0.9in}|>{\centering\arraybackslash}m{4.2in}|}
\hline
 &  & $\hat{Y}_1,\hat{Y}_2,\hat{Y}_3,\hat{Y}_4,\hat{Y}_5,$ & $(M_{e_1}+\hat{X}_1,\,M_{e_2}+\hat{X}_2,\,M_{e_3}+\hat{X}_2,\,M_{e_4}+\hat{X}_3,\,M_{e_5}+\max\{\hat{X}_1,\hat{X}_4\},$ \\ \cline{3-4}
$R_{\text{all}}$ & $\hat{X}_1,\ldots,\hat{X}_{11}$ & $\hat{Y}_6,\hat{Y}_7,\hat{Y}_8,$ & $M_{e_6}+\max\{\hat{X}_2,\hat{X}_5\},\,M_{e_7}+\max\{\hat{X}_2,\hat{X}_5\},\,M_{e_8}+\max\{\hat{X}_2,\hat{X}_6\},$ \\ \cline{3-4}
 & & $\hat{Y}_9,\hat{Y}_{10},\hat{Y}_{11},$ & $M_{e_9}+\max\{\hat{X}_3,\hat{X}_7\},\,M_{e_{10}}+\max\{\hat{X}_1,\hat{X}_2,\hat{X}_4,\hat{X}_5,\hat{X}_8\},\,M_{e_{11}}+\max\{\hat{X}_2,\hat{X}_5,\hat{X}_9\},$ \\ \cline{3-4}
 & & $\hat{Y}_{12},\hat{Y}_{13}$ & $M_{e_{12}}+\max\{\hat{X}_1,\hat{X}_2,\hat{X}_4,\hat{X}_5,\hat{X}_8,\hat{X}_9,\hat{X}_{10}\},\,M_{e_{13}}+\max\{\hat{X}_2,\hat{X}_3,\hat{X}_6,\hat{X}_7,\hat{X}_{11}\})$ \\ \hline
\end{tabular}
\end{minipage}
\vspace{10pt}
\hrule
\end{table*}
The message vector is $(\hat{X},\hat{Y})=(\hat{X}_1,\ldots,\hat{X}_{11},\hat{Y}_{e_1},\ldots,\hat{Y}_{e_{13}})$ and each message is a $10$-bit word. A functional index code obtained using \eqref{eq_map} is given in Table~\ref{tab_M_e1}.
The decoding maps of clients obtained using \eqref{eq_dh_E},\eqref{eq_dh_T}, and \eqref{eq_dh_all} are given in the last column of Table~\ref{tab_decmap1} (by $M_e$ we mean $M_e(\hat{X},\hat{Y})$ for all $e\in E$). \hfill $\square$
\end{example}

\begin{table*}[t] 
\caption{}
\label{tab_decmap2}
\scriptsize
\centering
\renewcommand{\tabcolsep}{4pt}
\renewcommand{\arraystretch}{1.4}
\begin{tabular}{|c|c|c|c|}
\hline
\textit{Client} &\textit{Has-set} &\textit{Want-set}& \textit{Decoding Map} \\ \hline \hline
$R_{e_1}$ & $\hat{X}_1,\hat{X}_2$ & $\hat{Y}_{e_1}$ & $M_{e_1}+\hat{X}_1$ \\ \hline
$R_{e_2}$ & $\hat{X}_1,\hat{X}_2$ & $\hat{Y}_{e_2}$ & $M_{e_2}+\hat{X}_2$ \\ \hline
$R_{e_3}$ & $\hat{X}_1,\hat{X}_2$ & $\hat{Y}_{e_3}$ & $M_{e_3}+\hat{X}_1+\hat{X}_2$ \\ \hline
$R_{e_4}$ & $\hat{X}_3,\hat{X}_4$ & $\hat{Y}_{e_4}$ & $M_{e_4}+\hat{X}_3+\hat{X}_4$ \\ \hline
$R_{e_5}$ & $\hat{X}_3,\hat{X}_4$ & $\hat{Y}_{e_5}$ & $M_{e_5}+\hat{X}_4$ \\ \hline
$R_{e_6}$ & $\hat{X}_3,\hat{X}_4$ & $\hat{Y}_{e_6}$ & $M_{e_6}+\hat{X}_3$ \\ \hline
$R_{e_7}$ & $\hat{Y}_{e_1}$ & $\hat{Y}_{e_7}$ & $M_{e_7}+M_{e_1}+\hat{Y}_{e_1}$ \\ \hline
$R_{e_8}$ & $\hat{Y}_{e_1}$ & $\hat{Y}_{e_8}$ & $M_{e_8}+M_{e_1}+\hat{Y}_{e_1}$ \\ \hline
$R_{e_9}$ & $\hat{Y}_{e_1}$ & $\hat{Y}_{e_9}$ & $M_{e_9}+M_{e_1}+\hat{Y}_{e_1}$ \\ \hline
$R_{e_{10}}$ & $\hat{Y}_{e_2},\hat{Y}_{e_4}$ & $\hat{Y}_{e_{10}}$ & $M_{e_{10}}+M_{e_2}+\hat{Y}_{e_2}+M_{e_4}+\hat{Y}_{e_4}$ \\ \hline
$R_{e_{11}}$ & $\hat{Y}_{e_2},\hat{Y}_{e_4}$ & $\hat{Y}_{e_{11}}$ & $M_{e_{11}}+M_{e_2}+\hat{Y}_{e_2}+M_{e_4}+\hat{Y}_{e_4}$ \\ \hline
$R_{e_{12}}$ & $\hat{Y}_{e_2},\hat{Y}_{e_4}$ & $\hat{Y}_{e_{12}}$ & $M_{e_{12}}+M_{e_2}+\hat{Y}_{e_2}+M_{e_4}+\hat{Y}_{e_4}$ \\ \hline
\end{tabular}
\qquad
\begin{tabular}{|c|c|c|c|}
\hline
\textit{Client} &\textit{Has-set} &\textit{Want-set}& \textit{Decoding Map} \\ \hline \hline
$R_{e_{13}}$ & $\hat{Y}_{e_3},\hat{Y}_{e_5}$ & $\hat{Y}_{e_{13}}$ & $M_{e_{13}}+M_{e_3}+\hat{Y}_{e_3}+M_{e_5}+\hat{Y}_{e_5}$ \\ \hline
$R_{e_{14}}$ & $\hat{Y}_{e_3},\hat{Y}_{e_5}$ & $\hat{Y}_{e_{14}}$ & $M_{e_{14}}+M_{e_3}+\hat{Y}_{e_3}+M_{e_5}+\hat{Y}_{e_5}$ \\ \hline
$R_{e_{15}}$ & $\hat{Y}_{e_3},\hat{Y}_{e_5}$ & $\hat{Y}_{e_{15}}$ & $M_{e_{15}}+M_{e_3}+\hat{Y}_{e_3}+M_{e_5}+\hat{Y}_{e_5}$ \\ \hline
$R_{e_{16}}$ & $\hat{Y}_{e_6}$ & $\hat{Y}_{e_{16}}$ & $M_{e_{16}}+M_{e_6}+\hat{Y}_{e_6}$ \\ \hline
$R_{e_{17}}$ & $\hat{Y}_{e_6}$ & $\hat{Y}_{e_{17}}$ & $M_{e_{17}}+M_{e_6}+\hat{Y}_{e_6}$ \\ \hline
$R_{e_{18}}$ & $\hat{Y}_{e_6}$ & $\hat{Y}_{e_{18}}$ & $M_{e_{18}}+M_{e_6}+\hat{Y}_{e_6}$ \\ \hline
$R_{t_1}$ & $\hat{Y}_{e_7},\hat{Y}_{e_{10}}$ & $\hat{X}_1+\hat{X}_2+\hat{X}_3+\hat{X}_4$ & $M_{e_7}+\hat{Y}_{e_7}+M_{e_{10}}+\hat{Y}_{e_{10}}$ \\ \hline
$R_{t_2}$ & $\hat{Y}_{e_8},\hat{Y}_{e_{13}}$ & $\hat{X}_2+\hat{X}_4$ & $M_{e_8}+\hat{Y}_{e_8}+M_{e_{13}}+\hat{Y}_{e_{13}}$ \\ \hline
$R_{t_3}$ & $\hat{Y}_{e_9},\hat{Y}_{e_{16}}$ & $\hat{X}_1+\hat{X}_3$ & $M_{e_9}+\hat{Y}_{e_9}+M_{e_{16}}+\hat{Y}_{e_{16}}$ \\ \hline
$R_{t_4}$ & $\hat{Y}_{e_{11}},\hat{Y}_{e_{14}}$ & $\hat{X}_1+\hat{X}_3$ & $M_{e_{11}}+\hat{Y}_{e_{11}}+M_{e_{14}}+\hat{Y}_{e_{14}}$ \\ \hline
$R_{t_5}$ & $\hat{Y}_{e_{12}},\hat{Y}_{e_{17}}$ & $\hat{X}_2+\hat{X}_4$ & $M_{e_{12}}+\hat{Y}_{e_{12}}+M_{e_{17}}+\hat{Y}_{e_{17}}$ \\ \hline
$R_{t_6}$ & $\hat{Y}_{e_{15}},\hat{Y}_{e_{18}}$ & $\hat{X}_1+\hat{X}_2+\hat{X}_3+\hat{X}_4$ & $M_{e_{15}}+\hat{Y}_{e_{15}}+M_{e_{18}}+\hat{Y}_{e_{18}}$ \\ \hline
\end{tabular}
\vspace{5pt} \\
\begin{tabular}{|c|c|c|c|}
\hline
\textit{Client} &\textit{Has-set} &\textit{Want-set}& \textit{Decoding Map} \\ \hline \hline
 & $\hat{X}_1,\hat{X}_2,$ & $\hat{Y}_1,\ldots,\hat{Y}_{6}$& $(M_{e_1}+\hat{X}_1, M_{e_2}+\hat{X}_2, M_{e_3}+\hat{X}_1+\hat{X}_2, M_{e_4}+\hat{X}_3+\hat{X}_4, M_{e_5}+\hat{X}_4, M_{e_6}+\hat{X}_3,$ \\ \cline{3-4}
$R_{\text{all}}$ & $\hat{X}_3,\hat{X}_4$ & $\hat{Y}_7,\ldots,\hat{Y}_{12}$& $M_{e_7}+\hat{X}_1, M_{e_8}+\hat{X}_1, M_{e_9}+\hat{X}_1, M_{e_{10}}+\hat{X}_2+\hat{X}_3+\hat{X}_4, M_{e_{11}}+\hat{X}_2+\hat{X}_3+\hat{X}_4, M_{e_{12}}+\hat{X}_2+\hat{X}_3+\hat{X}_4,$ \\ \cline{3-4}
 &  & $\hat{Y}_{13},\ldots,\hat{Y}_{18}$& $M_{e_{13}}+\hat{X}_1+\hat{X}_2+\hat{X}_4, M_{e_{14}}+\hat{X}_1+\hat{X}_2+\hat{X}_4, M_{e_{15}}+\hat{X}_1+\hat{X}_2+\hat{X}_4, M_{e_{16}}+\hat{X}_3, M_{e_{17}}+\hat{X}_3, M_{e_{18}}+\hat{X}_3)$ \\ \hline
\end{tabular}
\end{table*}
\begin{example}
Consider the network computation problem given in Fig.~\ref{fig_fnc2}. Each message takes value from the binary field and each edge has a capacity of $1$ bit. The corresponding functional index coding problem obtained using Definition~\ref{nc_2_fic} is given by the first three columns of Table~\ref{tab_decmap2}. The message vector is $(\hat{X},\hat{Y})=(\hat{X}_1,\ldots,\hat{X}_4,\hat{Y}_{e_1},\ldots,\hat{Y}_{e_{18}})$ and each message takes value from the binary field. A functional index code for the same is given in Table~\ref{tab_M_e2}, and the decoding maps of clients obtained using \eqref{eq_dh_E},\eqref{eq_dh_T}, and \eqref{eq_dh_all} are given in the last column of Table~\ref{tab_decmap2}.

\begin{figure}[h]
\centering
\includegraphics[scale=0.60]{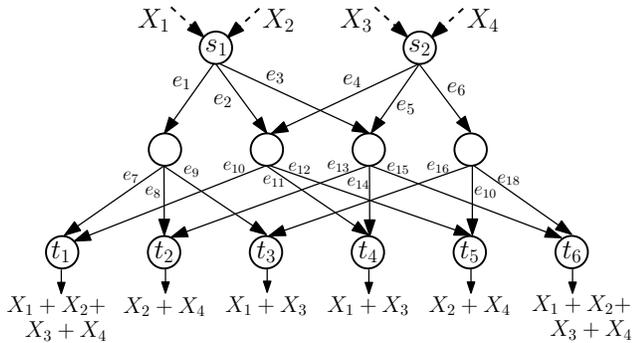}
\caption{A function computation problem.}
\label{fig_fnc2}
\vspace{-10pt}
\end{figure}

\begin{table}[h]
\caption{}
\label{tab_M_e2}
\scriptsize
\centering
\renewcommand{\tabcolsep}{4pt}
\renewcommand{\arraystretch}{1.4}
\begin{tabular}{|c|c|}\hline
$M_e(\hat{X},\hat{Y})$ & $\hat{Y}_e+F_e(\hat{X})$ \\ \hline \hline
$M_{e_1}$ & $\hat{Y}_{e_1}+\hat{X}_1$ \\ \hline 
$M_{e_2}$ & $\hat{Y}_{e_2}+\hat{X}_2$ \\ \hline 
$M_{e_3}$ & $\hat{Y}_{e_3}+\hat{X}_1+\hat{X}_2$ \\ \hline 
$M_{e_4}$ & $\hat{Y}_{e_4}+\hat{X}_3+\hat{X}_4$ \\ \hline 
$M_{e_5}$ & $\hat{Y}_{e_5}+\hat{X}_4$ \\ \hline 
$M_{e_6}$ & $\hat{Y}_{e_6}+\hat{X}_3$ \\ \hline 
$M_{e_7}$ & $\hat{Y}_{e_7}+\hat{X}_1$ \\ \hline 
$M_{e_8}$ & $\hat{Y}_{e_8}+\hat{X}_1$ \\ \hline 
$M_{e_9}$ & $\hat{Y}_{e_9}+\hat{X}_1$ \\ \hline
\end{tabular} 
\quad
\begin{tabular}{|c|c|}\hline
$M_e(\hat{X},\hat{Y})$ & $\hat{Y}_e+F_e(\hat{X})$ \\ \hline \hline
$M_{e_{10}}$ & $\hat{Y}_{e_{10}}+\hat{X}_2+\hat{X}_3+\hat{X}_4$ \\ \hline 
$M_{e_{11}}$ & $\hat{Y}_{e_{11}}+\hat{X}_2+\hat{X}_3+\hat{X}_4$ \\ \hline 
$M_{e_{12}}$ & $\hat{Y}_{e_{12}}+\hat{X}_2+\hat{X}_3+\hat{X}_4$ \\ \hline 
$M_{e_{13}}$ & $\hat{Y}_{e_{13}}+\hat{X}_1+\hat{X}_2+\hat{X}_4$ \\ \hline 
$M_{e_{14}}$ & $\hat{Y}_{e_{14}}+\hat{X}_1+\hat{X}_2+\hat{X}_4$ \\ \hline 
$M_{e_{15}}$ & $\hat{Y}_{e_{15}}+\hat{X}_1+\hat{X}_2+\hat{X}_4$ \\ \hline 
$M_{e_{16}}$ & $\hat{Y}_{e_{16}}+\hat{X}_3$ \\ \hline 
$M_{e_{17}}$ & $\hat{Y}_{e_{17}}+\hat{X}_3$ \\ \hline 
$M_{e_{18}}$ & $\hat{Y}_{e_{18}}+\hat{X}_3$ \\ \hline 
\end{tabular}
\end{table}

The local encoding kernels obtained using \eqref{eq_dhre2fe} (by setting $M_e=0$ for all $e\in E$ and replacing $\hat{X}_k$ with $X_k$ and $\hat{Y}_e$ with $Y_e$ in the last column of Table~\ref{tab_decmap2}) and the resulting global encoding kernels (by \eqref{eq_f_2_F}) are given in Table~\ref{tab_ic2nc}.
\begin{table}[h]
\caption{}
\label{tab_ic2nc}
\scriptsize
\centering
\renewcommand{\tabcolsep}{4pt}
\renewcommand{\arraystretch}{1.4}
\begin{tabular}{|c|c|c|}\hline
$Y_{e}$ & $f_e$ & $F_e$ \\ \hline \hline
$Y_{e_1}$ & $X_1$ & $X_1$ \\ \hline
$Y_{e_2}$ & $X_2$ & $X_2$ \\ \hline
$Y_{e_3}$ & $X_1+X_2$ & $X_1+X_2$ \\ \hline
$Y_{e_4}$ & $X_3+X_4$ & $X_3+X_4$ \\ \hline
$Y_{e_5}$ & $X_4$ & $X_4$ \\ \hline
$Y_{e_6}$ & $X_3$ & $X_3$ \\ \hline
$Y_{e_7}$ & $Y_{e_1}$ & $X_1$ \\ \hline
$Y_{e_8}$ & $Y_{e_1}$ & $X_1$ \\ \hline
$Y_{e_9}$ & $Y_{e_1}$ & $X_1$ \\ \hline
\end{tabular}
\qquad
\begin{tabular}{|c|c|c|}\hline
$Y_e$ & $f_e$ & $F_e$ \\ \hline \hline
$Y_{e_{10}}$ & $Y_{e_2}+Y_{e_4}$ & $X_2+X_3+X_4$ \\ \hline
$Y_{e_{11}}$ & $Y_{e_2}+Y_{e_4}$ & $X_2+X_3+X_4$ \\ \hline
$Y_{e_{12}}$ & $Y_{e_2}+Y_{e_4}$ & $X_2+X_3+X_4$ \\ \hline
$Y_{e_{13}}$ & $Y_{e_3}+Y_{e_5}$ & $X_1+X_2+X_4$ \\ \hline
$Y_{e_{14}}$ & $Y_{e_3}+Y_{e_5}$ & $X_1+X_2+X_4$ \\ \hline
$Y_{e_{15}}$ & $Y_{e_3}+Y_{e_5}$ & $X_1+X_2+X_4$ \\ \hline
$Y_{e_{16}}$ & $Y_{e_6}$ & $X_3$ \\ \hline
$Y_{e_{17}}$ & $Y_{e_6}$ & $X_3$ \\ \hline
$Y_{e_{18}}$ & $Y_{e_6}$ & $X_3$ \\ \hline
\end{tabular}
\end{table}

In Table~\ref{tab_ic2dec} below, the decoding maps for sinks $t_1,\ldots,t_6$ obtained using \eqref{eq_dhrt2dt} (by setting $M_e=0$ for all $e\in E$ and replacing $\hat{Y}_e$ with $Y_e$ in the last column of Table~\ref{tab_decmap2}) are verified to output desired function at each sink by substituting values of $Y_e$, $e\in E$, from the last column (global encoding kernels) of Table~\ref{tab_ic2nc}.
\begin{table}[h]
\caption{}
\label{tab_ic2dec}
\scriptsize
\centering
\renewcommand{\tabcolsep}{4pt}
\renewcommand{\arraystretch}{1.4}
\begin{tabular}{|c|c|}\hline
\textit{Sink} & \textit{Decoding Map} \\ \hline \hline
$t_1$ & $Y_{e_7}+Y_{e_{10}}=X_1+(X_2+X_3+X_4)$ \\ \hline
$t_2$ & $Y_{e_8}+Y_{e_{13}}=X_1+(X_1+X_2+X_4)=X_2+X_4$ \\ \hline
$t_3$ & $Y_{e_9}+Y_{e_{16}}=X_1+X_3$ \\ \hline
$t_4$ & $Y_{e_{11}}+Y_{e_{14}}=(X_2+X_3+X_4)+(X_1+X_2+X_4)=X_1+X_3$ \\ \hline
$t_5$ & $Y_{e_{12}}+Y_{e_{17}}=(X_2+X_3+X_4)+X_3=X_2+X_4$ \\ \hline
$t_6$ & $Y_{e_{15}}+Y_{e_{18}}=(X_1+X_2+X_4)+X_3$ \\ \hline
\end{tabular}
\end{table}

Thus, Table~\ref{tab_ic2nc} is a network code for network computation problem of Fig.~\ref{fig_fnc2}. \hfill $\square$
\end{example}

\section{From Functional Index Coding Problem to Network Computation Problem}\label{sec_fic2fnc}
In the previous section, we showed that any network computation problem can be converted into a functional index coding problem wherein only the \textit{Want-sets} of some of the receivers included functions of source messages. In this section, we obtain a network computation problem from a given functional index coding problem with only \textit{Want-sets} containing functions of messages and show that solution to one problem can be converted into a solution for the other.

\begin{definition}\label{fic_2_fnc}
Let $\mathcal{I}(Z,\mathcal{R})$ be a functional index coding problem with $K$ messages $Z_1,Z_2,\ldots,Z_K$, where, for every $k\in [K]$, $Z_k$ is uniformly distributed over $\mathbb{F}_q^{n_k}$ for some positive integer $n_k$ and $\mathcal{R}=\{R_1,R_2,\ldots,R_M\}$ is the set of $M$ clients. For any client $R_i$, its \textit{Has-set} $H_i$ is a subset of the messages set $\{Z_1,Z_2,\ldots,Z_K\}$ and its \textit{Want-set} $W_i$ includes some functions of the messages, $W_i=\{w_{i,1}(Z),w_{i,2}(Z),\ldots ,w_{i,|W_i|}(Z)\}$, where $N_K=n_1+\ldots+n_K$ and $w_{i,l}:\mathbb{F}_q^{N_K} \rightarrow \mathbb{F}_q$ for $1\leqslant l\leqslant |W_i|$. That is, only the \textit{Want-sets} of the clients include functions of messages. Let $\ell$ be a positive integer. A network computation problem $\mathcal{F_I}(\mathcal{N}(S\cup T\cup B,\tilde{E}\cup E),X,\{G_t:t\in T\})$ can be constructed as follows:
\begin{enumerate}
\item The message vector $X=Z$, i.e., $K$ messages $X_1,X_2,\ldots,X_K$ are generated in the network and $X_k=Z_k$ for all $k\in [K]$.
\item The set of vertices is $V=S \cup T \cup B$. 
\begin{itemize}
\item $S=\{s_1,s_2,\ldots,s_K\}$ is the set of $K$ source vertices. Source $s_k$ generates message $X_k$.
\item $T=\{t_1,t_2,\ldots,t_M\}$ is the set of $M$ sink nodes. Sink $t_m$ demands the functions in the set $G_{t_m}=W_m$.
\item $B=\{v_B ,v'_B\}$.
\end{itemize}
\item The set $\tilde{E}$ is the set of tailless source edges, where $\tilde{E}=\{\tilde{e}_1,\tilde{e}_2,\ldots,\tilde{e}_K\}$ ($\tilde{e}_k$ terminates at $s_k$, corresponds to $X_k$, and has capacity $n_k$) and the set of directed network edges is $E=E_1\cup E_2\cup E_3 \cup \{e_B\}$.
\begin{itemize}
\item $E_1=\{(s_k,t_m):Z_k\in H_m,s_k\in S,t_m\in T\}$. The capacity of an edge $(s_k,t_m)\in E_1$ is $n_k$. 
\item $E_2=\{(s_k,v_B):s_k\in S\}$ and the capacity of an edge $(s_k,v_B)$ is $n_k$.
\item $E_3=\{(v'_B,t_m):t_m\in T\}$ and the capacity of an edge $(v'_B,t_m)$ is $\ell$.
\item $e_B=(v_B,v'_B)$ and the capacity of this edge is $\ell$.
\end{itemize}
\end{enumerate}
\end{definition}

\begin{remark} \label{rem_2}
\begin{enumerate}[a)]
\item Since the demands of the sink $t_m$ in $\mathcal{F_I}$ and the client $R_m$ in $\mathcal{I}$ are the same, $G_{t_m}(\mathbf{x})=W_m(\mathbf{x})$ for every $\mathbf{x}\in \mathbb{F}_q^{N_K}$.
\item Node $v_B$ is the only node in the network that has in-degree greater than one and hence the only node that can perform coding operations. The node $v_B$ has access to all the source messages $X_1,X_2,\ldots,X_K$ via the edges in $E_2$ and each sink $t_m\in T$ has access to all the information received by $v'_B$ via the edges in $E_3$. Thus, the node $v_B$ and the coded message it passes on the edges $(v_B,v'_B)$ in the network computation problem correspond, respectively, to the transmitter (that knows all the messages $Z_1,Z_2,\ldots,Z_K$ and performs coding operations) and the functional index code it broadcasts (received by all the clients) in the functional index coding problem. 
\item The node $v'_B$ simply forwards the coded message it receives from $v_B$ to all the sinks $t_m\in T$ via the edges in $E_3$.
\item The edges in $E_1$ represent the \textit{Has-sets} of the clients; there is an edge from $s_k$ to $t_m$ in the network if and only if client $R_m$ knows $Z_k$. Also, $In(t_m)=\{(s_k,t_m):Z_k\in H_m\}\cup \{(v'_B,t_m)\}$ and $|In(t_m)|=|H_m|+1$ for every $t_m\in T$ (or, equivalently $R_m\in\mathcal{R}$). In other words, through an edge $(s_k,t_m)\in E_1$, $t_m$ receives $X_k$, i.e.,
\begin{align} \label{eq_in_t_H}
(X_k)_{k:(s_k,t_m)\in In(t_m)}=H_m(X),
\end{align} 
and through edge $(v'_B,t_m)$ it receives the coded message on edge $e_B$.
\end{enumerate}
\end{remark}

\begin{proposition} \label{prop_2}
For a given functional index coding problem $\mathcal{I}(Z,\mathcal{R})$ (with only \textit{Want-sets} containing functions of messages), let $\mathcal{F_I}(\mathcal{N}(S\cup T\cup B,\tilde{E}\cup E),X,\{G_t:t\in T\})$ be the corresponding network computation problem constructed using Definition~\ref{fic_2_fnc}. Let $\ell$ be a positive integer. Then, a functional index code of length $\ell$ for $\mathcal{I}$ exists if and only if a network code for $\mathcal{F_I}$ exists.
\end{proposition}
\begin{proof}
Let $N_K=n_1+n_2+\ldots+n_K$. We will show that any network code for $\mathcal{F_I}$ can be converted into a functional index code of length $\ell$ for $\mathcal{I}$ and vice versa.

\textit{Converting a network code into a functional index code:} Suppose there is a network code for $\mathcal{F_I}$ and let $F_{e_B}$ be the global encoding kernel of the edge $e_B$ (since this is the only edge that carries coded message) and $\{D_{t_m}:t_m\in T\}$ be the set of decoding functions of sinks. Then by \eqref{eq_dec_F} and Remark~\ref{rem_2}(d), for every sink $t_m\in T$ we have
\begin{align} \label{eq_a}
D_{t_m}\left(F_{e_B}(X),(X_k)_{k:(s_k,t_m)\in In(t_m)}\right)=G_{t_m}(X).
\end{align}
By \eqref{eq_a} and \eqref{eq_in_t_H}
\begin{align} \nonumber
D_{t_m}&\left(F_{e_B}(X),(X_k)_{k:(s_k,t_m)\in In(t_m)}\right)=\qquad\qquad\qquad\\ 
&\qquad\qquad D_{t_m}(F_{e_B}(X),H_m(X))=G_{t_m(X)}. \label{eq_b}
\end{align}
Define a map $\mathcal{M}:\mathbb{F}_q^{N_K}\rightarrow \mathbb{F}_q^\ell$ by 
\begin{align}\label{eq_map_nc2fic}
\mathcal{M}(Z)=F_{e_B}(Z)
\end{align}
and a map $\hat{D}_{R_m}$ for every client $R_m\in \mathcal{R}$ by
\begin{align}\label{eq_decmap_nc2fic}
\hat{D}_{R_m}(\mathcal{M}(Z),H_m(Z))=D_{t_m}(\mathcal{M}(Z),H_m(Z)).
\end{align}
For any receiver $R_m\in \mathcal{R}$, by \eqref{eq_decmap_nc2fic}, \eqref{eq_map_nc2fic}, \eqref{eq_b}, and Remark~\ref{rem_2}(a), we have 
\begin{align*}
\hat{D}_{R_m}(\mathcal{M}(Z),H_m(Z))&=D_{t_m}(\mathcal{M}(Z),H_m(Z))\\
&=D_{t_m}(F_{e_B}(Z),H_m(Z))=G_{t_m}(Z)\\
&=W_m(Z)
\end{align*}
and hence $\hat{D}_{R_m}$ satisfy \eqref{eq_Dhat}. Thus, \eqref{eq_map_nc2fic} defines a functional index code for $\mathcal{I}(Z,\mathcal{R})$ with decoding maps $\hat{D}_{R_m}$, $R_m\in\mathcal{R}$, as defined in \eqref{eq_decmap_nc2fic}.

\textit{Converting a functional index code into a network code:} Assume that a functional index code of length $\ell$ is given for $\mathcal{I}(Z,\mathcal{R})$, i.e., an encoding map $\mathcal{M}$ and decoding functions $\{\hat{D}_{R_m}:R_m\in \mathcal{R}\}$ are specified. We will construct a network code for the network computation problem $\mathcal{F_I}$, i.e., specify the global encoding kernel of the edge $e_B$ (since this the only edge that carries coded message) and decoding functions $D_{t_m}$ for each sink in $t_m\in T$, using the given functional index code.

By \eqref{eq_Dhat}, for every client $R_m$ we have
\begin{align} \label{eq_c}
\hat{D}_{R_m}(\mathcal{M}(Z),H_m(Z))=W_m(Z).
\end{align}
Define a map $F_{e_B}:\mathbb{F}_q^{N_K}\rightarrow \mathbb{F}_q^\ell$ by
\begin{align} \label{eq_d}
F_{e_B}(X)=\mathcal{M}(X).
\end{align}
For every sink $t_m\in T$, define a map $D_{t_m}$ that takes data received on the incoming edges of the sink $t_m$ as input by
\begin{align} \nonumber
D_{t_m}&\left(F_{e_B}(X),(X_k)_{k:(s_k,t_m)\in In(t_m)}\right)=\qquad\qquad\qquad \\
&\qquad\qquad\hat{D}_{R_m}\left(F_{e_B}(X),(X_k)_{k:(s_k,t_m)\in In(t_m)}\right).
\label{eq_e}
\end{align}
By \eqref{eq_e}, \eqref{eq_in_t_H}, \eqref{eq_c}, and Remark~\ref{rem_2}(a)
\begin{align*}
D_{t_m}&\left(F_{e_B}(X),(X_k)_{k:(s_k,t_m)\in In(t_m)}\right)=\qquad\qquad\qquad \\
&\quad\hat{D}_{R_m}(\mathcal{M}(X),H_m(X))=W_m(X)=G_{t_m}(X).
\end{align*}
Thus, \eqref{eq_d} and \eqref{eq_e} define a network code for $\mathcal{F_I}$.
\end{proof}

\begin{example}
Consider the functional index coding problem given by the first three columns of Table~\ref{tab_fic3}. Each message takes value from a binary field. Majority function is denoted by $Maj$, $Maj(a,b,c)=ab+bc+ac$, $a,b,c\in \{0,1\}$. An optimal functional index code that satisfies all clients is $(C_1,C_2,C_3)=(X_1+X_6,\,X_3+X_4,\,X_2+X_5)$. The decoding function of each client is given in the last columns of Table~\ref{tab_fic3}.
\begin{table}[h]
\caption{} \label{tab_fic3}
\scriptsize
\centering
\renewcommand{\tabcolsep}{1pt}
\renewcommand{\arraystretch}{1.4}
\begin{tabular}{|c|c|c|c|}\hline
\textit{}&\textit{Has-set}&\textit{Want-set}&\textit{Decoding Map} \\ \hline
$R_1$ & $X_2,X_3$ & $Maj(X_1+X_6,X_2+X_3,X_4)$ & $Maj(C_1,X_2+X_3,C_2+X_3)$\\ \hline
$R_2$ & $X_4,X_5$ & $X_1+X_5+X_6,X_3$ & $(C_1+X_5,C_2+X_4)$\\ \hline
$R_3$ & $X_3,X_6$ & $X_2+X_5,X_1+X_3+X_6$ & $C_3,C_1+X_3$\\ \hline
$R_4$ & $X_1$ & $X_6$ & $C_1+X_1$ \\ \hline
\end{tabular}
\end{table}

The network computation problem obtained using Definition~\ref{fic_2_fnc} is given in Fig.~\ref{fig_fic3}. The capacity of each dash-dotted edge is $\ell (\geqslant 0)$.

\begin{figure}[h]
\centering
\includegraphics[scale=0.5]{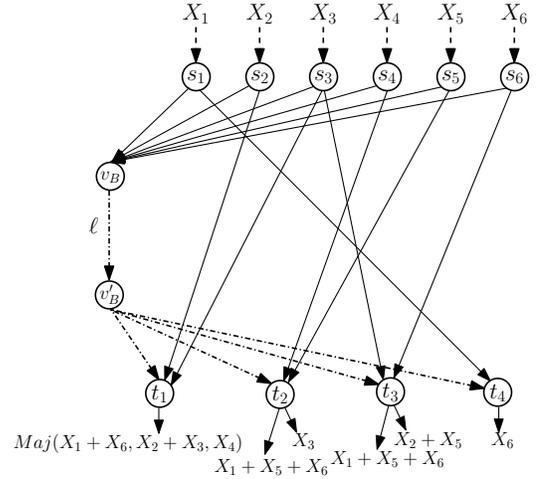}
\caption{Network computation problem corresponding to Table~\ref{tab_fic3}.}
\label{fig_fic3}
\end{figure}

By Proposition~\ref{prop_2}, a code for the network computation problem exists if and only if $\ell\geqslant 3$ since minimum $3$ transmissions are required for the given functional index coding problem. A network code that satisfies all sink demands is obtained by assigning $F_{e_B}(X_1,\ldots,X_6)=(X_1+X_6,\,X_3+X_4,\,X_2+X_5)$ as the global encoding kernel of edge $e_B=(v_B,v'_B)$. Upon receiving this coded message and the source messages on the incoming links, each sink can compute the value of its desired function.
\hfill $\square$
\end{example}

\section{Discussion}\label{sec_disc}
In this paper, we established a relationship between network computation problems and the class of functional index coding problems in which only the demands of the clients include functions of messages. We gave a method to convert any network computation problem into a functional index coding problem with some clients demanding functions of messages and proved that the network computation problem admits a solution if and only if all the clients' demands in the corresponding functional index coding problem can be satisfied in a specific number of transmissions (determined by the original network computation problem). We gave a method to convert a network code into a functional index code and vice versa. This means that algorithms to solve functional index coding problems can be used to obtain network codes for network computation problems. Next, given a functional index coding problem in which some clients demand functions of messages but the side information of each client is a subset of the message set known to the transmitter, we construct a network computation problem and show that a functional index code of a specified length for the former problem exists if and only if a network code for the latter problem exists. 

In this paper, we have considered only the functional index coding problem in which only the demands of the clients include functions and not the side information known to them \textit{a priori}. A direction of future work is to find what kind of networks and problems can be obtained from and converted to the functional index coding problems in which the side information can also include functions of messages known to the transmitter.


\section*{Acknowledgment}
This work was supported partly by the Science and Engineering Research Board (SERB) of Department of Science and Technology (DST), Government of India, through J.C. Bose National Fellowship to B. Sundar Rajan.



\end{document}